\documentclass{omasvjour3}

\smartqed  
\usepackage{graphicx}
\usepackage[OT1]{fontenc}
\usepackage{mathptmx}      

\usepackage{amssymb}
\usepackage{amsmath}
\usepackage{latexsym}
\usepackage{url}
\usepackage{float}
\newcommand{\xy}{X \rightarrow Y}

\newcommand{\ass}[2]{\ensuremath{{#1} \rightarrow {#2}}}

\spnewtheorem{observation}{Observation}{\bf}{\it}
\journalname{arXiv report}

\begin{document}

\title{Assessing the statistical significance of association rules}



\author{Wilhelmiina H{\"a}m{\"a}l{\"a}inen}


\institute{Department of Computer Science, University of Helsinki\\
              \email{whamalai@cs.helsinki.fi}           
}

\date{}

\maketitle

\begin{abstract}
An association rule is statistically significant, if it has a small
probability to occur by chance. It is well-known that the traditional
frequency-confidence framework does not produce statistically
significant rules. It can both accept spurious rules (type 1 error)
and reject significant rules (type 2 error). The same problem concerns
other commonly used interestingness measures and pruning
heuristics. 

In this paper, we inspect the most common measure functions --
frequency, confidence, degree of dependence, $\chi^2$, correlation
coefficient, and $J$-measure -- and redundancy reduction
techniques. For each technique, we analyze whether it can make type 1
or type 2 error and the conditions under which the error occurs. In
addition, we give new theoretical results which can be use to guide
the search for statistically significant association rules.
\end{abstract}

\section{Introduction}
\label{intro}

One of the most important tasks of data mining is the search for
partial dependencies in data. A partial dependency between attributes
$A_1,...,A_l$ means that some values $a_1,...,a_l$ occur together more
often than expected, if the attributes were independent. When two
attribute sets $X$ and $Y$ are partially dependent, the dependency can
be expressed as a rule $X=\overline{x} \rightarrow Y=\overline{y}$,
for the given value combinations $\overline{x}$ and $\overline{y}$. If
the rule is common and strong enough, it is called an {\em association
  rule} \cite{agrawalass}.

The commonness and strength of rule
$\ass{X=\overline{x}}{Y=\overline{y}}$ are measured by frequency\\
$fr(\ass{X=\overline{x}}{Y=\overline{y}})=P(X=\overline{x},Y=\overline{y})$
and confidence
$cf(\ass{X=\overline{x}}{Y=\overline{y}})=P(Y=\overline{y}|X=\overline{x})$. It
is required that $fr(\ass{X=\overline{x}}{Y=\overline{y}})\geq
min_{fr}$ and $cf(\ass{X=\overline{x}}{Y=\overline{y}})\geq min_{cf}$
for some user-defined thresholds $min_{fr}$ and $min_{cf}$.

The problem of traditional association rules is that they do not
reflect the statistical significance of partial
dependencies. Statistically, the significance of an association rule
is defined by the probability that it has occurred by chance. In
practice, the statistical significance depends on two things:
frequency, $P(X=\overline{x},Y=\overline{y})$, and the degree of
dependence,
$\frac{P(X=\overline{x},Y=\overline{y})}{P(X=\overline{x})P(Y=\overline{y})}$.
The stronger the dependence, the smaller the frequency can be, and
vice versa. This means that no absolute values can be given for the
minimum frequency and minimum confidence.

This problem is well-known, and especially Webb \cite{webb06,webbml}
has criticized the frequency-confidence framework. He has shown that
in the worst case, all discovered rules are spurious (statistically
insignificant). Berzal et al. \cite{berzal} and Aggarwal and Yu
\cite{aggarwalyu2} have shown that the frequency-confidence framework
is problematic, even if the absolute threshold values are not used.

Still the mainstream has ignored the problem, because the
efficiency of the search algorithms lies on the frequency-based
pruning. Nearly all search algorithms utilize the antimonotonicity of the
frequency: if set $X$ is not frequent (given some $min_{fr}$), then
none of its supersets $Y \supset X$ can be frequent. 

If the minimum frequency is set too high, several significant rules
can be missed. On the other hand, if the minimum frequency is too low,
a large number of spurious rules is accepted and the problem becomes
computationally intractable. In statistics, these two error types --
accepting spurious patterns and rejecting true patterns -- are known as
{\em type 1} and {\em type 2 errors}. 

As a solution, statistical or other objective measures have been used
to rank the discovered rules or to guide the search
(e.g.\ \cite{liuhsuma,morishitanakaya,morishitasese}). These measures
have generally two problems: either they are designed to assess
dependencies between attributes (not attribute values) and can miss
significant association rules, or they are statistically unsound.

In this paper, we examine how well the common interestingness measures and
search heuristics capture significant association rules. For each
method, we analyze whether it can make type 1 or type 2 error and
the conditions under the errors they occur. We give several important
results which can be used to select the most suitable search
heuristics for the given mining task. On the other hand, the new
theoretical results can give an insight for developing new,
statistically sound search methods for partial dependencies.

The paper is structured as follows: In Section 2 the problem of
finding statistically significant association rules is formalized.
The basic definitions of full and partial dependencies, association
rules, and statistical significance are given.

The effect of commonly used interest measures and search heuristics to the
statistical significance is analyzed in Sections 3--5.  In Section 3,  
we analyze how well the basic measures of association rules, frequency,
confidence, and the degree of dependency, indicate the statistical
significance. In Section 4, we proceed into more complex measures: $\chi^2$,
correlation coefficient $\phi$, and $J$-measure. The effect of common 
redundancy reduction techniques is analyzed in Section 5. 

The final conclusions are drawn in Section \ref{concl}.

\section{Definitions}
\label{definitions}

We begin by formalizing the problem and give exact definitions for
full and partial dependencies, association rules, and statistical
significance. The basic notations are defined in Table
\ref{reldb}. When it is clear from the context, we use abbreviations
$A$ and $\neg A$ for single attribute values $(A=1)$ and $(A=0)$, and
$X$ or $A_1,...,A_l$ for assignment $A_1=1,...,A_l=1$.

\begin{table}[!h]
\begin{center}
\caption{Basic notations.}
\footnotesize{
\begin{tabular}{|l|l|}
\hline
{\bf Notation}& {\bf Meaning}\\
\hline
$A,B,C$, $A_1,A_2,A_3,...$ & binary attributes (variables)\\
\hline
$a,b,c$, $a_i,a_2,a_3,... \in\{0,1\}$ & attribute values\\
\hline
$R=\{A_1,...,A_k\}$ & set of all attributes (relational schema)\\ 
\hline
$|R|=k$ & number of attributes in $R$\\
\hline
$Dom(R)=\{0,1\}^k$&attribute space (domain of $R$)\\
\hline 
$X,Y,Z \subseteq R$&attribute sets\\
\hline 
$|X|=|A_1,...,A_l|=l$&number of attributes in set $X$\\
\hline 
$Dom(X)=\{0,1\}^l \subseteq Dom(R)$&domain of $X$, $|X|=l$\\
\hline 
$(X=\overline{x})=\{(A_1=a_1),...,(A_l=a_l)\}$&event; attribute value 
assignment for $X$,\\ 
&$|X|=l$\\ 
\hline 
$t=\{A_1=t(A_1),...,A_k=t(A_k)\}$&row (tuple) according to $R$\\
\hline 
$r=\{t_1,...,t_n~|~t_i \in Dom(R)\}$&relation according to $R$\\ 
\hline
$|r|=n$& size of relation $r$ (the number of rows)\\
\hline
$\sigma_{X=\overline{x}}(r)=\{t \in r~|~t[X]=\overline{x}\}$&set of rows for which $X=\overline{x}$ holds\\
\hline
$m(X=\overline{x})=|\sigma_{X=\overline{x}}(r)|$&number of rows, for which $X=\overline{x}$ holds;\\ 
&$(X=\overline{x})$'s absolute frequency or support\\ 
\hline
$P(X=\overline{x})$&$(X=\overline{x})$'s relative frequency (probability)\\ 
&in $r$\\
\hline
$P(Y=\overline{y}|X=\overline{x})=\frac{P(X=\overline{x},Y=\overline{y})}{P(X=\overline{x})}$ & conditional probability of $Y$ given
$X$\\ 
&in $r$\\ 
\hline
\end{tabular}
}
\label{reldb}
\end{center}
\end{table}

\subsection{Statistical dependence}
\label{dependencies}

Statistical dependence is classically defined through statistical
independence (see e.g.\ \cite{silverstein,meo}). In the following, we will
concentrate on {\em two-way dependencies}, i.e.\ dependencies between
two attribute sets or events.

\begin{definition}[Statistical independence and dependence]
Let $X\subsetneq R$ and $Y\subseteq R\setminus X$ be sets of binary 
attributes. 

Events $X=\overline{x}$ and $Y=\overline{y}$, 
$\overline{x}\in Dom(X)$, $\overline{y}\in Dom(Y)$, are {\em mutually
independent}, if $P(X=\overline{x}, Y=\overline{y})=
P(X=\overline{x})P(Y=\overline{y})$.

Attribute sets $X$ and $Y$ are mutually independent, if 
$P(X=\overline{x}, Y=\overline{y})=
P(X=\overline{x})P(Y=\overline{y})$ for all 
value combinations $\overline{x}\in Dom(X)$ and $\overline{y}\in Dom(Y)$. 

If the events or attribute sets are not independent, they are 
{\em dependent}. 
\end{definition}

The following example demonstrates that attribute sets can be
dependent, even if some events are independent:

\begin{example}
Let $R=\{A,B,C\}$ be a set of binary attributes, where attribute $C$
depends on attribute set $\{A,B\}$.  Still it is possible that events
$(C=1)$ and $(A=1,B=1)$ are mutually independent. Table \ref{probesim}
gives an example of such a probability assignment.

\begin{table}[!h]
\begin{center}
\caption{A probability assignment, where attribute $C$ depends on set 
$\{A,B\}$, but event $(A,B,C)$ is independent. 
$0<d\leq \min\{P(A,\neg B)P(\neg C), 
(1-P(A,\neg B))P(C)\}$.}
\label{probesim}
\begin{tabular}{|l|l|}
\hline
$X$&$P(X)$\\
\hline
$ABC$&$P(A,B)P(C)$\\
\hline
$AB\neg C$&$P(A,B)P(\neg C)$\\
\hline
$A\neg BC$&$P(A,\neg B)P(C)+d$\\
\hline
$A\neg B\neg C$&$P(A,\neg B)P(\neg C)-d$\\
\hline
$\neg ABC$&$P(\neg A,B)P(C)-d$\\
\hline
$\neg AB\neg C$&$P(\neg A,B)P(\neg C)+d$\\
\hline
$\neg A\neg BC$&$P(\neg A,\neg B)P(C)$\\
\hline
$\neg A\neg B\neg C$&$P(\neg A,\neg B)P(\neg C)$\\
\hline
\end{tabular}
\end{center}
\end{table}

However, it is also possible that all events are dependent. An example
of such a probability assignment is given in Table \ref{probesim2}.

\begin{table}[!h]
\begin{center}
\caption{A probability assignment, where attribute $C$ depends on set 
$\{A,B\}$ and all events are dependent. 
$0<d\leq \min\{P(A,\neg B)P(\neg C),(1-P(A,\neg B))P(C)\}$.}
\label{probesim2}
\begin{tabular}{|l|l|}
\hline
$X$&$P(X)$\\
\hline
$ABC$&$P(A,B)P(C)-d$\\
\hline
$AB\neg C$&$P(A,B)P(\neg C)+d$\\
\hline
$A\neg BC$&$P(A,\neg B)P(C)+d$\\
\hline
$A\neg B\neg C$&$P(A,\neg B)P(\neg C)-d$\\
\hline
$\neg ABC$&$P(\neg A,B)P(C)-d$\\
\hline
$\neg AB\neg C$&$P(\neg A,B)P(\neg C)+d$\\
\hline
$\neg A\neg BC$&$P(\neg A,\neg B)P(C)+d$\\
\hline
$\neg A\neg B\neg C$&$P(\neg A,\neg B)P(\neg C)-d$\\
\hline
\end{tabular}
\end{center}
\end{table}

When we analyze the distributions further, we observe that in Table
\ref{probesim}, $C$ is actually dependent on $A$ and $B$ separately: 
$P(A,C)=P(A)P(C)+d$, $P(B,C)=P(B)P(C)-d$. In Table \ref{probesim2}, 
  $\{A,B\}$ is the minimal set which has a dependency with $C$. 
\end{example}

It is usually required that the dependency should be significant,
before events or attribute sets are called dependent. In the latter
case, this means that all value combinations $(X=\overline{x},
Y=\overline{y})$ should be represented in the data and the dependences
should be sufficiently strong for most events (e.g.\ 
\cite{brinmotwani,jaroszewicz}).

The strength of a statistical dependency between $(X=\overline{x})$
and $(Y=\overline{y})$ is defined by comparing $P(X=\overline{x},
Y=\overline{y})$ and $P(X=\overline{x})P(Y=\overline{y})$. 
The measure functions can be based on either the {\em
  absolute difference} ({\em dependence value} \cite{meo}),
$d(X=\overline{x},Y=\overline{y})=P(X=\overline{x},Y=\overline{y})-
P(X=\overline{x})P(Y=\overline{y})$, or the {\em relative difference}, 
$$r(X=\overline{x},Y=\overline{Y})=\frac{d(X=\overline{x},Y=\overline{y})}
{P(X=\overline{x})P(Y=\overline{y})}.$$

In the association rule literature, the relative difference is often
defined via another measure, called the {\em degree of dependence}
({\em dependence} \cite{wuzhangzhang}, {\em degree of independence}
\cite{yaozhong}, or {\em interest} \cite{brinmotwani}):

\begin{equation}
\gamma(X=\overline{x},Y=\overline{y})=
\frac{P(X=\overline{x},Y=\overline{y})}{P(X=\overline{x})P(Y=\overline{y})}=
1+\frac{d(X=\overline{x},Y=\overline{Y}}{P(X=\overline{x})P(Y=\overline{y})}.
\end{equation}

In the real world data, it is quite common that some value
combinations are overrepresented, while others are totally missing. In
this situation, we cannot make any judgements concerning dependences
between attribute sets, but still we can find significant dependencies
between certain events. In this paper, these two kinds of significant
dependencies are called partial and full dependencies: 

\begin{definition}[Partial and full dependence]
Let $X$ and $Y$ be like before. 
Attribute sets $X$ and $Y$ are called {\em partially
  dependent}, if the dependency between events 
$(X=\overline{x})$ and $(Y=\overline{y})$ is significant for some
$\overline{x} \in Dom(X)$ and $\overline{y}\in Dom(Y)$. 

$X$ and $Y$ are called {\em fully dependent}, if the dependency
between $X$ and $Y$ is significant. 
\end{definition}

Thus, full dependence implies partial dependence, but not vice
versa. This means that the methods for assessing the significance of 
full dependencies do not necessarily capture all significant partial
dependencies. 

One trick is to turn a partial dependency into a full dependency by
treating events $X=\overline{x}$ and $Y=\overline{y}$ as binary
attributes. Table \ref {XYdesim} gives a contingency table of the
associated probabilities. Now it is more likely that all four value
combinations are represented in the data and the methods for assessing
full dependencies can be applied.

\begin{table}[!h]
\begin{center}
\caption{A contingency table with probabilities of $P(X,Y)$, $P(X,\neg Y)$, 
$P(\neg X,Y)$ and $P(\neg X, \neg Y)$. If $d>0$, 
$d\leq \min\{P(\neg X)P(Y),P(X)P(\neg Y)\}$, and if $d<0$, $d\leq
  \min\{P(X)P(Y),P(\neg X)P(\neg Y)\}$.} 
\label{XYdesim}
\footnotesize{
\begin{tabular}{|l|l|l|l|}
\hline
& $Y$ & $\neg Y$ & $\Sigma$\\
\hline
$X$ & $P(X,Y)=$ & $P(X,\neg Y)=$ & $P(X)$\\
&$P(X)P(Y)+d$&$P(X)P(Y)-d$&\\
\hline
$\neg X$& $P(\neg X, Y)=$& $P(\neg X,\neg Y)=$&$P(\neg X)$\\
&$P(\neg X)P(Y)-d$&$P(\neg X)P(\neg Y)+d$&\\
\hline
$\Sigma$ & $P(Y)$ & $P(\neg Y)$ & $1$\\
\hline
\end{tabular}
}
\end{center}
\end{table}

\subsection{Association rules}
\label{association rules}

Often, the dependency between events is expressed as rule
$X=\overline{x} \rightarrow Y=\overline{y}$. Association rules
\cite{agrawalass} are a natural framework to express such
rules. Traditionally, association rules are defined in the {\em
  frequency-confidence framework}:

\begin{definition}[Association rule]
Let $R$ be a set of binary attributes and $r$ a relation according
to $R$.  Let $X \subsetneq R$ and $Y \subseteq R\setminus X$, be
attribute sets and $\overline{x} \in Dom(X)$ and $\overline{y} \in
Dom(Y)$ their value combinations.

The {\em confidence} of rule $(X=\overline{x}) \rightarrow (Y=\overline{y})$ is 
$$cf(X=\overline{x} \rightarrow
Y=\overline{y})=\frac{P(X=\overline{x},Y=\overline{y})}{P(X=\overline{x})}=P(Y=\overline{y}|X=\overline{x})$$
and the {\em frequency} of the rule is
$$fr(X=\overline{x} \rightarrow Y=\overline{y})=P(X=\overline{x},Y=\overline{y}).$$

Given user-defined thresholds $min_{cf}, min_{fr} \in [0,1]$, 
rule $(X=\overline{x}) \rightarrow (Y=\overline{y})$ is an {\em association
rule} in $r$, if 
\begin{itemize}
\item[(i)] $cf(X=\overline{x} \rightarrow Y=\overline{y})\geq min_{cf}$, and 

\item[(ii)] $fr(X=\overline{x} \rightarrow Y=\overline{y})\geq min_{fr}$.
\end{itemize}
\end{definition}

The first condition requires that an association rule should be strong
enough and the second condition requires that it should be common
enough. In this paper, we call rules association rules, even if no
thresholds $min_{fr}$ and $min_{cf}$ are specified.

Often it is assumed that the consequent $Y=\overline{y}$ contains just
one attribute, $|Y|=1$. When the consequent is a fixed class attribute
$C$, rules $X=\overline{x} \rightarrow C=c$, $c\in Dom(C)$, are called
{\em classification rules}.

Another common restriction is to allow only positive attribute
values ($A_i=1$). The reasons are mostly practical: in the traditional context
of market-basket data, most of the items do not occur in a single
basket. Thus, it is sensible to search only correlations between items
that often occur together. On the other hand, the number of items is
very large, typically $>1000$, and searching all association rules would be
impossible. In the other contexts, negative attribute values cannot be
excluded. For example, when we search dependencies in the demographic
data, we canot exclude all women, unmarried, employed,
etc.     

The main problem of the frequency-confidence framework is that the
minimum frequency and confidence requirements do not guarantee any
statistical dependence or significance
\cite{brinmotwani,aggarwalyu2,morishitasese}. However, most
researchers have adopted Piatetsky-Shapiro's \cite{piatetskyshapiro}
argument that a rule cannot be interesting, if its antecedent and
consequent are statistically independent. That is why it is often
demanded that $\gamma(X=\overline{x} \rightarrow Y=\overline{y})\neq
1$ (e.g.\ \cite{brinmotwani,wuzhangzhang,tankumar2}).  According to
the sign of $\gamma-1$, the rule or its type is called positive,
negative or independent (''null association rule'')
\cite{liuhsuma,garriga}. Usually, only positive dependencies are
searched, since they can be used for prediction.

We note that from the statistical point of view, the direction of a
rule ($\ass{X=\overline{x}}{Y=\overline{y}}$ or
$\ass{Y=\overline{y}}{X=\overline{x}}$) is a matter of choice. In the worst 
case, the direction can be misleading, since rules are usually
associated with causation and association rules (or correlations) do
not necessarily represent any causality relationship \cite{jermaine}.

Another important notice is that the association rules are not
implications. Especially, rule $\ass{Y}{X}$ is not the same as
$\ass{\neg X}{\neg Y}$. Unless $P(X)=P(Y)=0.5$, rules $Y \rightarrow
X$ and $\neg X \rightarrow \neg Y$ have different
frequencies, confidences and degrees of dependence.

\subsection{Statistical significance of partial dependencies}
\label{significance}

The idea of statistical significance tests is to estimate the
probability of the observed or a rarer phenomenon, under some null
hypothesis. When the objective is to test the significance of
the dependency between $X=\overline{x}$ and $Y=\overline{y}$, the null
hypothesis is the independence assumption: 
$P(X=\overline{x},Y=\overline{y})=P(X=\overline{x})P(Y=\overline{y})$. If
the estimated probability $p$ is very small, we can reject the
independence assumption, and assume that the observed dependency is
not due to chance, but significant at level $p$. The smaller $p$ is,
the more significant the observation is.

Usually the minimum requirement for any significance is
$p\leq 0.05$. It means that there is 5\% chance that a spurious rule
passes the significance test (``type 1 error''). If we test 10 000 rules,
it is likely that will find 500 spurious rules. This so called {\em
  multiple testing problem} is inherent in the knowledge discovery,
where we often perform an exhaustive search over all possible patterns. 

As a solution, the more patterns we test, the stricter bounds for the
significance we should use. The most well-known method is {\em
  Bonferroni adjustment} \cite{bonferroni}, where the desired
significance level $p$ is divided by the number of tests. In the
association rule discovery, we can give an upper bound for the number
of rules to be tested. However, this rule is so strict that there is
a risk that we do not recognize all significant patterns (``type 2
error''). Webb \cite{webbml,webb06} has argued that this is a less
serious problem than finding spurious rules, because the number of
rules is anyway large. He has also suggested another approach, where a
part of the data is held as an evaluation set. Now the number of rules
to be tested is known before testing, and higher significance levels
can be used.

Let us now analyze the significance of partial dependency
$X=\overline{x} \rightarrow Y=\overline{y}$. To simplify the
notations, the sets are denoted by $X$ and $Y$.

The significance of the observed frequency $m(X,Y)$ can be estimated
exactly by the binomial distribution. Each row in relation $r$,
$|r|=n$, corresponds to an independent Bernoulli trial, whose outcome
is either 1 ($XY$ occurs) or 0 ($XY$ does not occur). All rows are
mutually independent.

Assuming the independence of attributes $X$ and $Y$, combination $XY$
occurs on a row with probability $P(X)P(Y)$. Now the number of rows
containing $X,Y$ is a binomial random variable $M$ with parameters
$P(X)P(Y)$ and $n$. The mean of $M$ is $\mu_M=nP(X)P(Y)$
and its variance is $\sigma_M^2=nP(X)P(Y)(1-P(X)P(Y))$. The probability
that $M \geq m(X,Y)$ is

\begin{equation}
\label{bintn}
p=P(M\geq m(X,Y))= \sum_{i=m(X,Y)}^n \left({n \atop
  i}\right) (P(X)P(Y))^i (1-P(X)P(Y))^{n-i}.
\end{equation}

This can be approximated by the standard normal distribution 
$$p \approx 1-\Phi(t),$$
where $\Phi(t(X,Y))=\frac{1}{\sqrt{2\pi}}\int_{- \infty}^{t(X,Y)}
e^{-u^2/2}du$ is the standard normal cumulative distribution function
and $t(X,Y)$ is standardized $m(X,Y)$:

\begin{equation}
\label{tXY}
t(X,Y)=\frac{m(X,Y)-\mu_M}{\sigma_M}=\frac{m(X,Y)-nP(X)P(Y)}
{\sqrt{nP(X)P(Y)(1-P(X)P(Y))}}.
\end{equation}

The approximation is quite good for large $n$, but it should not be
used, when the expected counts $nP(X)P(Y)$ and $n(1-P(X)P(Y))$ are small. As a
rule of thumb, it is often required that $nP(X)P(Y)>5$ and
$n(1-P(X)P(Y))>5$ (e.g.\ \cite[p. 121]{milton}). 

The cumulative distribution function $\Phi(t)$ is quite difficult to
calculate, but for the association rule mining it is enough to know
$t(X,Y)$. Since $\Phi(t)$ is monotonically increasing, probability $p$
is monotonically decreasing in the terms of $t(X,Y)$. Thus, we can use
$t$ as a measure function for ranking association rules according to
their significance. On the other hand, we know that in the normal
distribution $P(-2\sigma_M<M-\mu_M<2\sigma_M)\approx 0.95$ or,
equivalently,

$$P\left(-2<\frac{M-\mu_M}{\sigma_M}<2\right)\approx 0.95.$$

I.e.\ $P(t(X,Y)\geq 2) \approx 0.025$, which is a minimum requirement
for any significance. Thus, we can prune all rules
$\ass{X}{Y}$ for which $t(X,Y)<2$. Generally, we can set the threshold
$K$ according to Chebyshev's inequality (the proof is given 
e.g.\ in \cite[pp. 780-781]{milton}):

$$P\left(-K<\frac{M-\mu_M}{\sigma_M}<K\right)\geq 1-\frac{1}{K^2}.$$

I.e.\ $P(t \geq K)<\frac{1}{2K^2}.$ Now the Bonferroni adjustment is
achieved by using $\sqrt{m}K$ instead of $K$, where $m$ is the number
of tests.

Equations (\ref{bintn}) and (\ref{tXY}) can be directly
generalized to attribute-value sets $\{A_1=a_1,\ldots,A_l=a_l\}$. Now
the null hypothesis is that all attributes are mutually independent:
$$P(A_1=a_1,\ldots,A_l=a_l)=P(A_1=a_1)P(A_2=a_2)\ldots P(A_l=a_l)=\Pi_{i=1}^l
P(A_i=a_i).$$
The significance of the dependence in set $\{A_1=a_1,\ldots,A_l=a_l\}$ is
measured by 
$$t(A_1=a_1,\ldots,A_l=a_l)=
\frac{m(A_1=a_1,\ldots,A_l=a_l)-n\Pi_{i=1}^l P(A_i=a_i)}
{\sqrt{n\Pi_{i=1}^l P(A_i=a_i)(1-\Pi_{i=1}^l P(A_i=a_i))}}.$$

\section{Basic measures for association rules}
\label{basicmeasures}

The statistical significance of rule $\ass{X}{Y}$ is a function of
$P(X)$, $P(Y)$ and $P(X,Y)$ (Equation (\ref{tXY})). All basic measures,
like frequency, confidence, and the degree of dependency, are composed
from these elements. In the frequency-confidence framework, the
assumption is that a high frequency ($P(X,Y)$) and a high confidence
($P(Y|X)$) indicate an interesting rule. In the following, we will
analyze conditions under which this assumption fails. As an
alternative, we analyze ``frequency-dependence framework'', and show that a
high frequency and a high degree of dependence, $\gamma$, indicate
statistical significance.

\subsection{Frequency and confidence}
\label{frcfsubsec}

Figure \ref{tfrcf} illustrates the significance of rule $X \rightarrow
Y$ as a function of frequency $P(X,Y)$ and confidence $P(Y|X)$, when
$Y$ is fixed. The values of $P(Y)$ are $0.2$, $0.4$, $0.6$ and
$0.8$. Now the significance measure $t$ is expressed as

$$\hat{t}=\frac{\sqrt{P(X,Y)}(P(Y|X)-P(Y))}{\sqrt{P(Y)(P(Y|X)-P(X,Y)P(Y))}}.$$

Data size $n$ is omitted, and the real significance is $t=\sqrt{n}\hat{t}$.  
The function is not defined when $P(Y|X)\leq P(X,Y)P(Y)$. For clarity, only 
areas where $t>0$ are drawn. In addition, it holds that $P(X,Y)\leq P(Y)$. 

\begin{figure*}[!h]
\begin{center}
\includegraphics[width=0.47\textwidth]{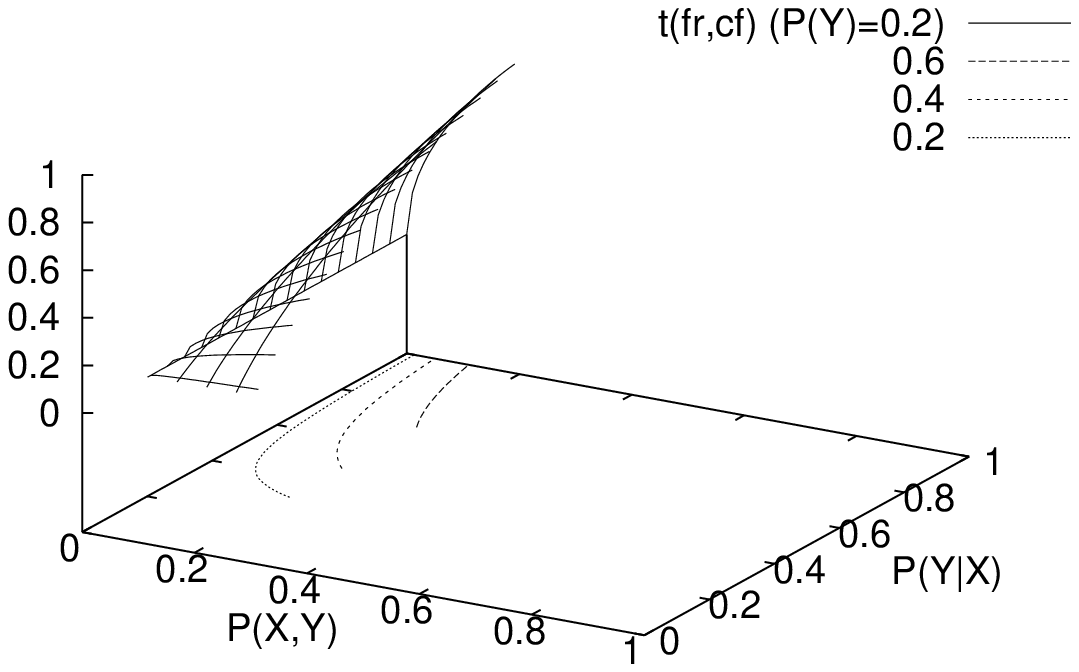}
\hfill
\includegraphics[width=0.47\textwidth]{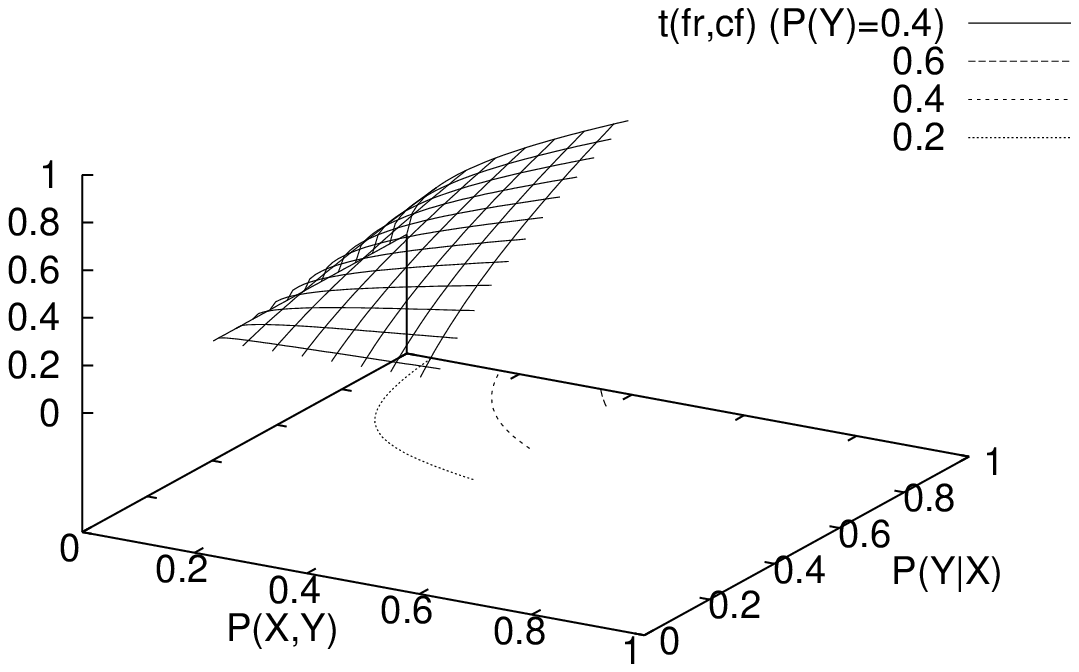}
\includegraphics[width=0.47\textwidth]{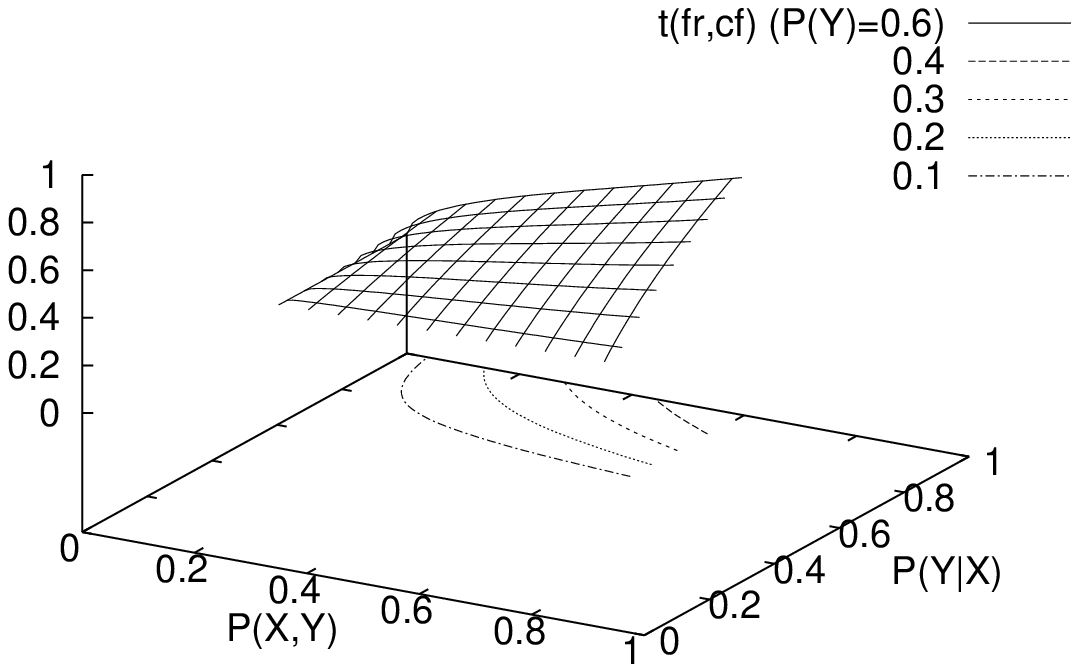}
\hfill
\includegraphics[width=0.47\textwidth]{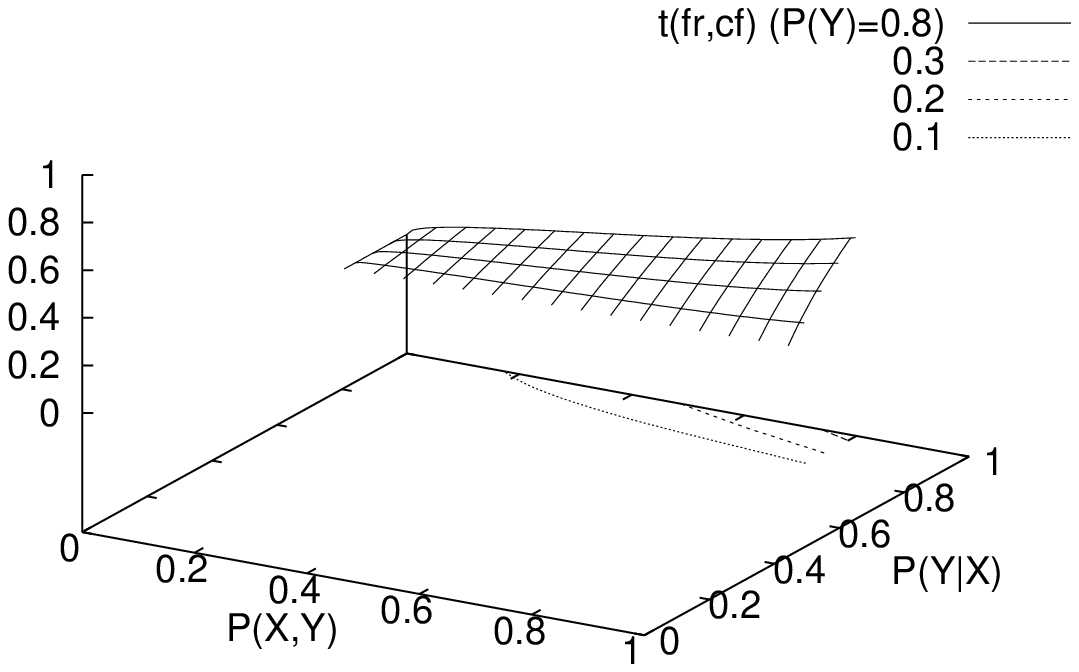}
\caption{The significance of $\ass{X}{Y}$ as a function of frequency 
$P(X,Y)$ and confidence $P(Y|X)$, when $P(Y)=0.2$ (left top), $P(Y)=0.4$ 
(right top), $P(Y)=0.6$ (left bottom) and $P(Y)=0.8$ (right bottom).}
\label{tfrcf}
\end{center}
\end{figure*}

The contours are compatible with a common intuition that the significance
is maximal, when both frequency and confidence are maximal. However,
the significance of the rule depends on $P(Y)$. The higher $P(Y)$ is,
the higher confidence the rule should have. The minimum requirement
for the confidence is $P(Y|X)>P(Y)$, since otherwise $t\leq 0$. In the
extreme case, when $P(Y)=1$, the rule is totally insignificant
($P(Y|X)=1$ for all $X$ and $t=0$). That is why rules with
different consequents are not comparable, in the terms of their
frequency and confidence. Often a rule with higher frequency and
confidence may be less significant than a weaker and less frequent
rule. A significant rule can be easily missed, when absolute
$min_{fr}$ and $min_{cf}$ values are used.

Generally, the preference for high frequency and confidence can cause
both type 1 and type 2 errors.  Let us first analyze what kind of
rules are accepted in the frequency-confidence framework. Let $X$ and
$Y$ be like in Table \ref{XYdesim}. The frequency of rule
$X=\overline{x} \rightarrow Y=\overline{y}$ is
$P(X=\overline{x})P(Y=\overline{y})+d$. Now any combination
$X=\overline{x}, Y=\overline{y}$ can be frequent, if
$P(X=\overline{x},Y=\overline{y})\geq min_{fr}$. If
$P(X=\overline{x})P(Y=\overline{y})\geq min_{fr}$, $X=\overline{x}$
and $Y=\overline{y}$ can be statistically independent ($d=0$) or even
negatively correlated ($d<0$).

The confidence of rule $X=\overline{x} \rightarrow Y=\overline{y}$ is
$P(Y=\overline{y})+\frac{d}{P(X=\overline{x})}$. The highest
confidence is achieved, when $P(Y=\overline{y})$ is large and
$P(X=\overline{x})$ is small. If $P(Y=\overline{y})\geq min_{cf}$, the
rule is confident, even if $X=\overline{x}$ and $Y=\overline{y}$ are
statistically independent.

On the other hand, the frequency-confidence framework can reject
significant rules. Let us analyze what the minimum frequency and
confidence should be for a rule to be significant. 

Let $t(\ass{X=\overline{x}}{Y=\overline{y}})\geq K$. This holds, when
the frequency is 
\begin{multline*}
P(X=\overline{x},Y=\overline{y})\geq \\
P(X=\overline{x})P(Y=\overline{y})+
\frac{K\sqrt{P(X=\overline{x})P(Y=\overline{y})(1-P(X=\overline{x})P(Y=\overline{y}))}}{n}
\end{multline*}
and the confidence is
\begin{multline*} 
P(Y=\overline{y}|X=\overline{x})\geq 
P(Y=\overline{y})+\frac{K\sqrt{P(Y=\overline{y})(1-P(X=\overline{x})P(Y=\overline{y}))}}{nP(X=\overline{x})}.
\end{multline*}

We see that the larger $n$ is, the smaller frequency and confidence
suffice for significance. On the other hand, the larger significance
level we require (expressed by $K$), the larger frequency and
confidence should be. The problem is that both of them depend on
$P(X=\overline{x})$ and $P(Y=\overline{y})$. The minimum frequency is an increasing function of
$P(X=\overline{x})P(Y=\overline{y})$. The minimum confidence is obtained from the minimum
frequency by dividing it by $P(X=\overline{x})$. Thus, the larger $P(X=\overline{x})$ is, the
larger $min_{cf}$ should be.

\begin{example}
Let $P(X)=P(Y)=0.5$ and $n=10 000$. Now rule $\ass{X}{Y}$ is
significant, if $P(X,Y)=0.25+\frac{K\sqrt{3}}{400}$ and 
$P(Y|X)=0.5+\frac{K\sqrt{3}}{200}$. Especially the confidence is low
and the rule is easily rejected with the typical $min_{cf}$ settings. 
For example, if we require that $K=10$ (indicating quite high 
significance), then confidence $0.5+\frac{10\sqrt{3}}{200}<0.60$
suffices.  
\end{example}

\subsection{Frequency and degree of dependence}
\label{frdegsubsec}

The problems of the frequency-confidence framework could be easily
corrected by using the degree of dependency,
$\gamma(X=\overline{x},Y=\overline{y})$, instead of confidence. This
approach is adopted e.g.\ in
\cite{srikantagrawal,tankumar2,aggarwalyu2}. Since
$\gamma(X=\overline{x},Y=\overline{y})=\frac{cf(X=\overline{x}
  \rightarrow Y=\overline{y}}{P(Y=\overline{y})}$, the frequency and
degree of dependence alone determine the statistical significance $t$.

Figure \ref{tfrg} illustrates the significance of rule
$\ass{X}{Y}$ as a function of frequency $P(X,Y)$ and degree of 
dependence $\gamma=\gamma(X \rightarrow Y)$:

$$\hat{t}(X \rightarrow
Y=\frac{\sqrt{P(X,Y)}(\gamma-1)}{\sqrt{\gamma-P(X,Y)}}.$$

Once again, the data size $n$ is omitted and $t=\sqrt{n}\hat{t}$.

\begin{figure*}[!h]
\begin{center}
\includegraphics[width=0.47\textwidth]{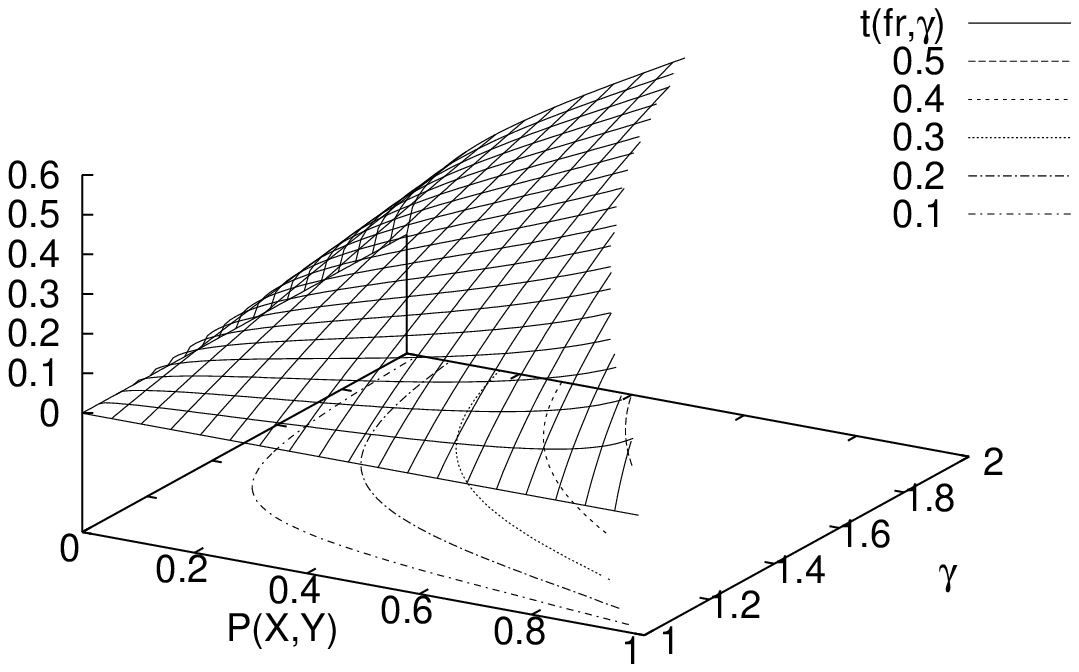}
\includegraphics[width=0.47\textwidth]{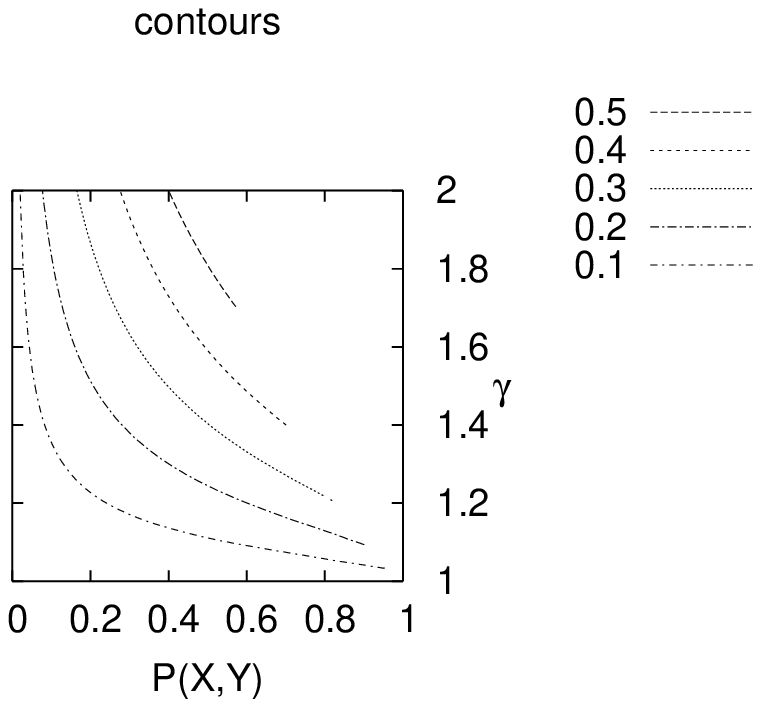}
\caption{The significance of $\ass{X}{Y}$ as a function of frequency 
$P(X,Y)$ and degree of dependency $\gamma$ (left) and the corresponding 
contours (right).}
\label{tfrg}
\end{center}
\end{figure*}

There is only one graph, because $P(Y)$ has no effect on
the significance, when $\gamma$ is given. $P(Y)$ determines only the
maximal possible value for $\gamma$: $\gamma(X \rightarrow Y)\leq
\frac{1}{P(Y)}$. (In the Figure, $\gamma \leq 2$ and $P(Y)\geq 0.5$.)
The minimum value, $\gamma\geq 1$, is set, because we are interested
in only positive correlations. 

The function is defined when $P(X,Y)\gamma \leq 1$, because 
$\gamma \leq \frac{1}{\max\{{P(X)},P(Y)\}}$ and $P(X,Y) \gamma \leq P(X)\gamma 
\leq 1$. 

From the contours we observe that $t$ is nearly symmetric in the
terms of $P(X,Y)$ and $(\gamma-1)$. It means that the larger the
frequency is, the smaller the degree of dependence can be, and vice
versa. If rule $R_1$ has both higher frequency and higher degree of dependence 
than rule $R_2$, it is more significant. If $R_1$ has only higher
frequency, then the $t$-values of rules should be compared to decide
the significance order.

The following theorem expresses the relationship between the frequency 
and the degree of dependence:

\begin{theorem}
When $t(\ass{X}{Y})=K$, 
$$P(X,Y)=\frac{K^2\gamma}{n(\gamma-1)^2+K^2}.$$
\end{theorem}

\begin{proof}
By solving 
$$t(\ass{X}{Y})=\frac{\sqrt{nP(X)}(\gamma-1)}{\sqrt{\gamma-P(X)}}=K.$$
\hfill $\Box$ 
\end{proof}

This result can be used for pruning areas in the search space, when an 
upper bound for $\gamma$ is known. At least areas where $K<2$ can be 
safely pruned, since $t\geq 2$ is a minimum requirement for any significance. 

The simplest method to search all statistically significant rules is
to search all frequent sets with sufficiently small $min_{fr}$ and
then select from each frequent set the rules with sufficient $t$. 

The following theorem gives a safe minimum frequency threshold
for the whole data set. It guarantees that no significant rules are missed. 
For simplicity, we assume that $|Y|=1$.

\begin{theorem}
\label{safeminfr}
Let $p_{min}=\min\{P(A_i=a_i)~|~A_i \in R, a_i=\{0,1\}\}$. Let $K\geq 2$ be 
the desired significance level.  
For all sets $X \subseteq R$ and any $A \in X$
\begin{itemize}
\item[(i)] $\gamma(\ass{X\setminus A}{A})\leq \frac{1}{p_{min}}$ and 
\item[(ii)] $\ass{X}{A}$ cannot be significant, unless
$$P(X)\geq \frac{K^2p_{min}}{n(1-p_{min})^2+K^2p_{min}^2}.$$
\end{itemize}
\end{theorem}

\begin{proof}
By solving 
$$t(\ass{X\setminus
  A}{A})=\frac{\sqrt{nP(X)}(\gamma-1)}{\sqrt{\gamma-P(X)}} \geq K.$$
\hfill $\Box$ 
\end{proof}

\section{Measures for dependence detection}
\label{measuressec}

Next, we analyze the most common objective measures for
dependence detection. We focus on the $\chi^2$-measure, which is the most
common statistical method for assessing the significance of
dependencies. It is often used in association rule mining, when the
goal is to find statistically significant association rules. Pearson
correlation coefficient $\phi$ is another statistical measure, which
has been applied to association rule mining. 

$J$-measure \cite{smythgoodman} is derived from the mutual information, which is
an information-theoretic measure for assessing dependencies between
attributes. It is especially designed for ranking decision rules, and
often applied in the association rule discovery.

Empirical comparisons of these and other interestingness measures can be 
found in e.g.\ \cite{vilalta,tankumar,tankumar2}.

\subsection{$\chi^2$-measure}
\label{chi2measure}

\subsubsection{Definition}
\label{chi2definition}

The $\chi^2$-{\em independence test} is the most popular statistical
test for detecting dependencies between attributes.
The idea of the $\chi^2$ test is to compare the observed
frequencies $O(m(X))$ to the expected frequencies
$E(m(X))$ by 
$$\chi^2(X)=\sum_{\overline{x} \in Dom(X)}
\frac{O((m(X=\overline{x}))-E(m(X=\overline{x})))^2}{E(m(X=\overline{x}))}.$$

When the test variable is approximately normally distributed, the test
measure follows the $\chi^2$-distribution. Usually this assumption holds
for large $n$. As a rule of thumb, it is suggested
(e.g.\ \cite[p. 630]{milton}) that all of the expected frequencies should be
at least 5.

When we test a dependency between two attribute sets, $X$ and $Y$, the
contingency table contains only four cells (Table
\ref{XYdesim}). Now the test metric is  
\begin{multline*}
\chi^2(X,Y)=\sum_{i=0}^1 \sum_{j=0}^1
\frac{(m(X=i,Y=j)-nP(X=i)P(Y=j))^2}{nP(X=i)P(Y=j)}=\\
\frac{n(P(X=1,Y=1)-P(X=1)P(Y=1))^2}{P(X=1)P(X=0)P(Y=1)P(Y=0)}.\\
\end{multline*}

If $\chi^2(X,Y)$ is less than the critical $\chi^2$ value at level $p$ and 1
degree of freedom, $X$ and $Y$ are statistically independent with
probability $1-p$. Otherwise, the dependency is significant at level
$p$. 

The above equation can be generalized to measure dependencies between
all variables in set $X=A_1,\ldots,A_l$:

\begin{multline*}
\chi^2(X)=\sum_{(a_1,\ldots,a_l)} 
\frac{n(P(A_1=a_1,\ldots,A_l=a_l)-P(A_1=a_1)\ldots P(A_l=a_l))^2}
{P(A_1=a_1)\ldots P(A_l=a_l)}.
\end{multline*}

\subsubsection{Applying $\chi^2$ in the association rule discovery}
\label{applyingchi2}

The simplest way to use $\chi^2$-measure in the association rule discovery
is to generate rules from frequent sets based on their
$\chi^2$-values. For each frequent set $X$, all rules of form
$X\setminus Y \rightarrow Y$ with a sufficient $\chi^2$-value are
selected (e.g.\ \cite{dong}).

This approach does not find all rules which are significant in the
$\chi^2$ sense. First, the rules are preselected according to their
frequency. If the minimum frequency is set too high, some significant
rules are missed.  

Second, it is possible that a weak rule ($P(Y|X)\leq 0.5$) is
selected, because its companion rules $\ass{X}{\neg Y}$, $\ass{\neg
  X}{Y}$, and/or $\ass{\neg X}{\neg Y}$ are significant. The rule
confidence can be used to check that $P(Y|X)>P(\neg Y|X)$, but it does
not guarantee that $\ass{X}{Y}$ is significant. As a solution, it is
often required (e.g.\ \cite{liuhsuma}) that
$P(X,Y)>P(X)P(Y)$. Unfortunately, it is still possible that the high
$\chi^2$-value is due to $\ass{\neg X}{\neg Y}$.

The first problem would be solved, if we could search the rules
directly with the $\chi^2$-measure. Unfortunately, this is not
feasible, since $\chi^2$-measure is not monotonic. For any
rule $X \rightarrow Y$ and its generalization $Z \rightarrow Y$, $Z
\subseteq X$, it is unknown, whether
$\chi^2(\ass{X}{Y})>\chi^2(\ass{Z}{Y})$ or $\chi^2(\ass{X}{Y})\leq
\chi^2(\ass{Z}{Y})$. 

There are at least two solutions to this problem: First, 
$\chi^2$ is used to find only the interesting attribute sets \cite{brinmotwani}.
Second, the convexity of the $\chi^2$-measure can utilized in
searching optimal rules with a fixed consequent $C$ 
\cite{morishitanakaya,morishitasese}. 

When $\chi^2$ is calculated for attribute sets, it is upwards
closed under set inclusion. This means that the $\chi^2$ value
can only increase, when attributes are added to a set. (Proof 
for the special case $|Z|=2$, $|X|=3$, $Z\subsetneq X$, 
is given in \cite{brinmotwani}.)

\begin{theorem}
For all attribute sets $X,Z$, $Z \subsetneq X$, $\chi^2(Z)\leq \chi^2(X)$.\\
\end{theorem}

\begin{proof}
Let $X=Z,A$, $|Z|=l$ and $|Z|=l+1$. $\chi^2(Z)$ contains 
$2^l$ terms of form\\
$\frac{n(P(A_1=a_1,\ldots,A_l=a_l)-P(A_1=a_1)\ldots P(A_l=a_l))^2}
{P(A_1=a_1)\ldots P(A_l=a_l)}=\frac{n(U-V)^2}{V^2}$. For each such term
$\chi^2(X)$ contains two terms:
\begin{multline*}
\frac{n(P(A_1=a_1,\ldots,A_l=a_l,A_{l+1}=1)-P(A_1=a_1)\ldots P(A_l=a_l)P(A_{l+1}=1))^2}
{P(A_1=a_1)\ldots P(A_l=a_l)P(A_{l+1}=1)}+\\
\frac{n(P(A_1=a_1,\ldots,A_l=a_l,A_{l+1}=0)-P(A_1=a_1)\ldots P(A_l=a_l)P(A_{l+1}=0))^2}
{P(A_1=a_1)\ldots P(A_l=a_l)P(A_{l+1}=0)}=\\
\frac{n(UP(A|Z)-VP(A))^2}{P(A)}+\frac{n(UP(\neg A|Z)-VP(\neg
  A))^2}{P(\neg A)}.
\end{multline*}
Now it is enough to show that 
\begin{multline*}
\frac{(U-V)^2}{V^2}\leq
\frac{(UP(A|Z)-VP(A))^2}{P(A)}+\frac{(UP(\neg A|Z)-VP(\neg
  A))^2}{P(\neg A)}\\
\Leftrightarrow \frac{U^2-2UV+V^2}{V}\leq \frac{U^2(P(A|Z)^2+P(A)P(\neg
  A|Z))}{VP(A)P(\neg A)}+\frac{-2UV+V^2}{V}.
\end{multline*}
This is always true, because 
\begin{multline*}
U^2\leq \frac{U^2(P(A|Z)^2+p(A)P(\neg A|Z))}{P(A)P(\neg A)} ~ \Leftrightarrow
P(A|Z)^2-P(A)P(A|Z)+P(A)^2\geq 0. 
\end{multline*}
\hfill $\Box$ 
\end{proof}

Thus, the most significant sets are the most specific, containing all
$k$ attributes. That is why Brin et al.\ \cite{brinmotwani,silverstein} used
$\chi^2$-test to find the ''minimally correlated sets'', i.e.\ the
most general attribute sets $X$ such that $\chi^2(X)\geq min_{\chi^2}$
for some cutoff value $min_{\chi^2}$. The type of correlation in set
$X=A_1,\ldots,A_l$ was determined by the interest measure 
$$\rho(A_1,\ldots,A_l)=\frac{P(A_1,\ldots,A_l}{P(A_1)\ldots P(A_l))}.$$ 
In addition, they used a new frequency measure for pruning: 
$$fr(X)=\max \left\{s~\left|~
\frac{|\{(X=\overline{x})~|~x \in Dom(X) \wedge P(X=\overline{x})\geq s\}|}{|Dom(X)|}\geq p\right\}\right.,$$
where $p\in]0,1]$.
This measure demands that in the contingency table of $|Dom(X)|$ cells the
frequency must be sufficient in at least $p|Dom(X)|$ cells. High $p$ and
$min_{fr}$ values produce effective pruning, but several significant
attribute sets can be missed, if their absolute frequency is too low
or the dependency is only a partial dependency. For example, this
heuristic ignores a dependency if $P(X=\overline{x})$ is high for some
$\overline{x} \in Dom(X)$, but
$P(X=\overline{x}_1)=P(X=\overline{x}_2)=\ldots=P(X=\overline{x}_l)$
for all $\overline{x_i} \neq \overline{x}$. In addition, we note that
parameters $p$ and $min_{fr}$ determine how many attributes $X$ can
contain, since $|X|\leq -\log(p \cdot min_{fr})$. For example, if $p=0.25$
and $min_{fr}=0.01$ (as suggested in \cite{brinmotwani}), $|X|\leq 8$.

Brin et al.\ did not generate any rules, even if the attribute sets
were called ''correlation rules'' \cite{brinmotwani} or ''dependence
rules'' \cite{silverstein}. A natural question is whether we could
generate significant rules from the correlated sets. Unfortuantely,
the dependence in a set is only a sufficient and not necessary
condition for two-way dependencies. In addition, it is possible that
none of the association rules generated from a correlated set is
necessarily significant \cite{morishitanakaya}.

The second approach, introduced by Morishita et
al. \cite{morishitanakaya,morishitasese}, is to utilize the convexity
of the $\chi^2$ function, when the consequent $C$ is fixed. The idea
is to prune a branch containing rule $\ass{Z}{C}$ and all its
specialization rules $\ass{X}{C}$, $Z \subseteq X$, if 
$\max\{\chi^2(\ass{X}{C})\}<min_{\chi^2}$ for the given cutoff value
$min_{chi^2}$. Because $\chi^2$ is convex,
$\max\{\chi^2(\ass{X}{C})\}<min_{\chi^2}$ can be bounded by equation

$$\chi^2(X \rightarrow C) \leq \max\left\{\frac{nP(Z,C)P(\neg
  C)}{(1-P(Z,C))P(C)}, \frac{nP(Z,\neg C)P(C)}{(1-P(Z,\neg C))P(\neg C)}
\right\}.$$

Now the frequency-based pruning is not necessary and it is possible to
find all rules with a sufficient $\chi^2$-value or the best rules in
the $\chi^2$ sense. This approach works correctly, when the goal is
to find full dependencies. Partial dependencies with fixed $C$ could
be searched similarly by applying the properties of the $t$-measure.

\subsubsection{Analysis}
\label{analysischi2}

The main problem of the $\chi^2$-independence test is that it designed
to measure dependencies between attributes. That is why it can fail to
detect significant partial dependencies. On the other hand,
$\chi^2$-test can yield a high value, thus indicating a significant
dependency, even if the tested events were nearly
independent. Negative correlations can be pruned by an extra test,
$P(X,Y)>P(X)P(Y)$, but it does not guarantee that the high
$\chi^2$-value is due to $X \rightarrow Y$.

Let us analyze the $\chi^2$-value, when  $P(X,Y)=P(X)P(Y)+d$ (Table 
\ref{XYdesim}). Now $\chi^2$ can be defined in the terms of $d$: 

$$\chi^2(X,Y)=\frac{nd^2}{P(X)P(\neg X)P(Y)P(\neg Y)}.$$

$\chi^2$ is high, when $n$ and $|d|$ are large and $P(X)P(\neg X)P(Y)
P(\neg Y)$ is small. 
The minimum value ($16nd^2$) is achieved, when $P(X)=P(Y)=0.5$, and
the maximum, when $P(X)$ and $P(Y)$ approach either 0 or 1. 
For example, if $P(X)=P(Y)=0.01$, $\chi^2=10
000 nd^2$, and even minimal $d$ suffices. E.g.\ if $n=1000$, 
$d\geq 0.8\cdot 10^{-3}$ for level 0.01, and $P(X,Y)=0.0009$. 

The problem is that if $P(X)$ and/or $P(Y)$ are large, the relative
difference $\frac{d}{P(X)P(Y)}$ is small and the partial dependency
between $X$ and $Y$ is not significant. Still the $\chi^2$-value can be
large, because $\frac{d}{P(\neg X)P(\neg Y)}$ is large. Thus, the high
$\chi^2$-value is due to partial dependency $\ass{\neg X}{\neg Y}$,
and $\ass{X}{Y}$ is a false discovery (type 1 error).

\begin{example}
Let $P(X)=P(Y)=1-\epsilon$ for
arbitrary small $\epsilon>0$.  Let $d$ be maximal
i.e.\ $d=P(X)(1-P(Y))=(1-P(X))P(Y)=\epsilon(1-\epsilon)<\epsilon$.  
(The relative difference is still very small,
$\frac{d}{P(X)P(Y)}=\frac{\epsilon}{1-\epsilon}$.) Now $\chi^2(X,Y)$ is very
large, the same as the data size, $n$:

$$\chi^2=\frac{nd^2}{P(X)P(Y)(1-P(X))(1-P(Y))}=\frac{n\epsilon^2(1-\epsilon)^2}
{\epsilon^2(1-\epsilon)^2}=n.$$

Still, rule $\ass{X}{Y}$ is insignificant, since
$$t(\ass{X}{Y})=\frac{\sqrt{n}(1-\epsilon)\epsilon}
{(1-\epsilon)\sqrt{1-(1-\epsilon)^2}}=\frac{\sqrt{n\epsilon}}
{\sqrt{2-\epsilon}} \rightarrow 0,$$
when $\epsilon \rightarrow 0$. 

The high $\chi^2$-value is due to partial dependency $\ass{\neg X}{\neg Y}$, 
which has a high $t$-value:
$$t(\ass{\neg X}{\neg
  Y})=\frac{\sqrt{n(1-\epsilon)}}{\sqrt{1+\epsilon}} \rightarrow
\sqrt{n},$$

when $\epsilon \rightarrow 0$. 

Rules $\ass{X}{\neg Y}$ and $\ass{\neg X}{Y}$ are meaningless, with 
$$t=\frac{\sqrt{n\epsilon(1-\epsilon)}}{\sqrt{1-\epsilon+\epsilon^2}}
<\frac{\sqrt{n\epsilon(1-\epsilon)}}{\sqrt{1-\epsilon}}=\sqrt{n\epsilon} 
\rightarrow 0.$$
\end{example}

$\chi^2$-measure is less likely to cause type 2 errors, i.e.\ to
reject significant partial dependencies. The reason is that 
the $\chi^2$-value of rule $\ass{X}{Y}$ increases quadratically 
in the terms of its $t$-value: 
 
\begin{theorem}
If $t(X\rightarrow Y)=K$, then $\chi^2(X,Y)\geq K^2$.
\end{theorem}

\begin{proof}
Let $x=P(X)$ and $y=P(Y)$. 
If $t(X\rightarrow Y)=K$, then 
$$nd^2=K^2xy(1-xy) \textrm{ and } 
\chi^2(X,Y)=\frac{nd^2}{xy(1-x)(1-y)}=\frac{K^2(1-xy)}{(1-x)(1-y)}\geq K^2,$$
since $(1-x)(1-y)\leq 1-xy \mbox{ for all } x,y \in[0,1].$
\hfill $\Box$ 
\end{proof}

If an association rule is just sufficiently significant, it passes
also the $\chi^2$-test. However, the relative order of rules according
to their $\chi^2$-values does not reflect their actual significance. 
If only $m$ best rules are selected, it is possible that all of them
are spurious and all significant rules are rejected.

\subsection{Correlation coefficient}
\label{correlation}

Some authors (e.g.\ \cite{tankumar2}) have suggested
Pearson correlation coefficient $\phi$ to measure the significance of an
association rule.  Traditionally, the Pearson correlation coefficient is
used to measure linear dependencies between numeric attributes.  When
the Pearson correlation coefficient is calculated for the
binary attributes, it reduces to the square root of $\chi^2/n$:

$$\phi(X,Y)=\frac{P(X,Y)-P(X)P(Y)}{\sqrt{P(X)P(\neg X)P(Y)P(\neg Y)}}=
\sqrt{\frac{\chi^2(X,Y)}{n}}.$$

Like $\chi^2(X,Y)$, $\phi(X,Y)=0$, when $P(X,Y)=P(X)P(Y)$, and the
variables are mutually independent. Otherwise, the sign of $\phi$
tells whether the correlation is positive ($\phi>0$) or negative
($\phi<0$). 

The problem is to decide when the correlation is significant. General
guidelines are sometimes given for defining a weak, moderate, or strong
correlation, but they are rather arbitrary, because the significance
depends on the data size, $n$. The smaller $n$ is, the larger $\phi$
should be, to be statistically significant. That is why the correlation 
coefficient can produce very misleading results when applied to 
the association rule discovery. 

We will first show that a rule can be insignificant, even if 
correlation coefficient $\phi(X,Y)=1$.  This means that $\phi$-measure can 
produce false discoveries (type 1 error).

\begin{observation}
When $P(X)$ and $P(Y)$ approach 1, it is possible that $\phi(X,Y)=1$, 
even if $t(X,Y)<K$ for any $K>0$.
\end{observation}

\begin{proof} 
Let $P(X)=P(Y)=1-\epsilon$ for arbitrary small $\epsilon>0$. 
Let $d$ be maximal 
i.e.\  $d=P(X)(1-P(Y))=(1-P(X))P(Y)=\epsilon(1-\epsilon)$. 
Now the correlation coefficient is 1: 
$$\phi(X,Y)=\frac{d}{P(X)(1-P(X))P(Y)(1-P(X))}=
\frac{(1-\epsilon)\epsilon}{(1-\epsilon)\epsilon}=1.$$

Still, for any $K>0$, $t(X,Y)<K$: 
\begin{multline*}
t(X,Y)=\frac{\sqrt{n}(1-\epsilon)\epsilon}
{(1-\epsilon)\sqrt{1-(1-\epsilon)^2}}=\frac{\sqrt{n\epsilon}}
{\sqrt{2-\epsilon}} <K$ $\Leftrightarrow$ $\epsilon < \frac{2K^2}{n+K^2}.
\end{multline*}
\hfill $\Box$ 
\end{proof} 

On the other hand, it is possible that $\phi$-measure rejects
significant rules (type 2 error), especially when $n$ is large. The
following observation shows that this can happen, when $P(X)$ and
$P(Y)$ are relatively small. The smaller they are, the smaller $n$
suffices. Therefore, we recommend that the correlation coefficient
should be totally avoided as an interestingness measure for
association rules.

\begin{observation}
It is possible that $\phi(X,Y)\rightarrow 0$, when $n\rightarrow \infty$, 
even if rule $\ass{X}{Y}$ is significant. 
\end{observation}

\begin{proof} 
Let $t(X,Y)=\frac{\sqrt{n}d}{\sqrt{P(X)P(Y)(1-P(X)P(Y))}}=K$. Then
\begin{multline*}
d=\frac{K\sqrt{P(X)P(Y)(1-P(X)P(Y))}}{\sqrt{n}} \mbox{ \ and }
\phi(X,Y)=\\ \frac{K\sqrt{P(X)P(Y)(1-P(X)P(Y))}}{\sqrt{nP(X)P(Y)(1-P(X))
(1-P(Y))}}=\frac{K\sqrt{1-P(X)P(Y)}}{\sqrt{n(1-P(X))(1-P(Y))}}.
\end{multline*}
When $P(X)\leq p$ and $P(Y)\leq p$ for some $p<1$, 
$\phi(X,Y)=\frac{K\sqrt{1+p}}{\sqrt{n(1-p)}} \rightarrow 0$, when 
$n \rightarrow \infty$.
\hfill $\Box$ 
\end{proof}

\subsection{$J$-measure}
\label{jmeasure}

Several objective measures used in the association rule discovery are
adopted from the decision tree learning. A decision tree can be
represented as a set of decision rules $X=\overline{x} \rightarrow
C=c$, where $c \in Dom(C)$ is a class value. The measure functions can
be used both in the rule generation (tree expansion) and post-pruning
phases. In both cases, the objective is to estimate the impact of a
single attribute-value condition to the generalization accuracy
(i.e.\ how well the classifier performs outside the training set).

In the pruning phase, the test is as follows: 
If $M(X=\overline{x} \rightarrow C=c)\geq M(X=\overline{x},A=a \rightarrow C=c)$, for the
given measure function $M$, then condition $A=a$ can be pruned. This
test may look fully adequate for the association rule pruning, but there is
one crucial difference: in the classification, both $X=\overline{x} \rightarrow
C=c$ and $\neg (X=\overline{x}) \rightarrow \neg(C=c)$ should be accurate, while
for association rules it is enough that $X=\overline{x} \rightarrow C=c$ is
significant. This means that the measure functions for classification
rules are too restrictive for association rules, and significant
associations can be missed.  

As an example, we analyze {\em $J$-measure} \cite{smythgoodman}, which
is often used to assess the interestingness of association
rules. $J$-measure is an information-theoretic measure derived from
the mutual information. For decision rules $X \rightarrow C$, 
$J$-measure is defined as
$$J(C|X)=P(X,C)\log\frac{P(C|X)}{P(C)}+P(X,\neg C)\log\frac{P(\neg
  C|X)}{P(\neg C)}\in [0,\infty[.$$
The larger $J$ is, the more interesting the rule should be. On the
other hand, $J(X,C)=0$, when the variables $X$ and $C$ are mutually
independent (assuming that $P(X)>0$).

$J$-measure contains two terms from the mutual information, $MI$,
between variables $X$ and $C$: $MI(X,C)=J(C|X)+J(C|\neg X)$. Thus, it
measures the information gain in two rules, $X \rightarrow C$ and $X
\rightarrow \neg C$. Rule $\ass{X}{C}$ has a high $J$-value, if its
complement rule $\ass{X}{\neg C}$ has high confidence (type 1
error). In the extreme case, when $P(C|X)=0$, $J(C|X)=P(X)\log
\frac{1}{P(\neg C)}$.

Type 2 error (rejecting true discoveries) can also occur with a
suitable distribution. One reason is that $J$-measure omits $n$, which
is crucial for the statistical significance. 

It can be easily shown that $J(C|X) \rightarrow 0$, when $P(X,C)\rightarrow 0$
or $P(C)\rightarrow 1$. In the latter case, rule $\ass{X}{C}$ cannot
be significant, but it is possible that a rule is significant, even if
its frequency is relatively small:

\begin{example}
Let $P(C|X)=0.75$ and $P(C)=0.5$. Now $J(C|X)=$\\ 
$P(X)(0.75\log 3-0.25)\approx 0.94P(X)$ and
$t(\ass{X}{C})=\frac{\sqrt{nP(X)}}{2\sqrt{2-P(X)}}$. For example, when
$P(X)=0.25$, $t=\frac{\sqrt{n}}{2\sqrt{7}}$, which high, when $n$ is
high. Still $J(C|X)\approx 0.23$, which indicates that the rule is
uninteresting.
\end{example}

According to \cite{blanchard}, other information-theoretic measures
are equally problematic, since they are designed for classification
rules. In addition, the values are difficult to interpret, unless they
express absolute independence.

In Table \ref{measurecomp}, we give a summary of the analyzed
measures. For each measure, we report, whether it can produce type 1
or type 2 error and all rules which affect the measure in addition to
the actually measured rule. 

\begin{table}
\caption{Summary of measures $M$ for assessing association rules. The 
occurrence of type 1 (accepting spurious rules) and type 2 (rejecting 
significant rules) errors is indicated by $+$ (occurs) and $-$ (does 
not occur). In addition, all rules which contribute to $M(\ass{X}{Y})$ 
are listed. For all measures except $fr$\&$cf$, the antecedent and 
consequent of each rule can be switched.}
\label{measurecomp}
\begin{center}
\begin{tabular}{|l|c|c|l|}
\hline
$M$&Type 1&Type 2&Rules\\
&error&error&\\
\hline
$fr$\&$cf$&$+$&$+$&$\ass{X}{Y}$\\
\hline
$fr$\&$\gamma$&$-$&$-$&$\ass{X}{Y}$\\
\hline
$\chi^2$&$+$&$-$&$\ass{X}{Y}$, $\ass{\neg X}{Y}$,\\
&&&$\ass{X}{\neg Y}$, $\ass{\neg X}{\neg Y}$\\
\hline
$\phi$&$+$&$+$&$\ass{X}{Y}$, $\ass{\neg X}{\neg Y}$\\
\hline
$J$&$+$&$+$&$\ass{X}{Y}$, $\ass{X}{\neg Y}$\\
\hline
\end{tabular}
\end{center}
\end{table}

\section{Effect of redundancy reduction}
\label{redreduction}

A common goal in association rule discovery is to find the most
general rules (containing the minimal number of attributes) which
satisfy the given search criteria. There is no sense to output complex 
rules $X\rightarrow Y$, if their generalizations $Z \rightarrow Y$,
$Z\subsetneq X$ are at least equally significant. Generally, the goal
is to find {\em minimal} (or most general) {\em interesting rules}, 
and prune out {\em redundant rules} \cite{bastide}. 

\subsection{General definition}
\label{redreddefinition}

Generally, redundancy can be defined as follows:

\begin{definition}[Minimal and redundant rules]
Given some interestingness measure $M$, rule $\ass{X}{Y}$ is a minimal 
rule, if there does not exist any rule $\ass{X'}{Y'}$ such 
that $X' \subseteq X$, $Y \subseteq Y'$ and 
$M(\ass{X'}{Y'})\geq M(\ass{X}{Y})$. If the rule is not minimal, then it 
is redundant. 
\end{definition}

Measure $M$ can be $t$-measure, $J$-measure, $\chi^2$-measure, or any
function which increases with the interestingness. In the
traditional frequency-confidence-framework with minimum frequency and 
confidence thresholds, $M$ can be defined as
$$M(\ass{X}{Y})=
\left\{
\begin{array}{l l}
1&\mbox{when } fr(\ass{X}{Y})\geq min_{fr} \mbox{ and } cf(\ass{X}{Y})
\geq min_{cf},\\
0&\mbox{otherwise.}\\
\end{array}
\right.$$

The motivation for the redundancy reduction is two-fold: First, a
smaller set of general rules is easier to interpret than a large set
of complex and often overlapping rules. Second, the problem complexity
is reduced, because it is enough to find a small subset of all
interesting rules. Thus, it is possible at least in principle to
perform the search more efficiently.

In the previous research, redundancy has been defined in various
ways. An important distinction is whether the redundancy refers to the
{\em interestingness} of a rule or the {\em representation} of
rules. In the first case, a rule is considered redundant, if there are
more general rules which are at least equally interesting. Such a
redundant rule contains no new information and it can be pruned
out. In the second case, even an interesting rule is considered
redundant, if it (or its frequency and confidence) can be derived from
the other rules in the representation. Now the rule is not pruned out,
but it is not represented explicitly. Examples of such {\em condensed
  representations} \cite{condensed} are closed \cite{pasquier99}, free
\cite{boulicaut00}, and non-derivable sets \cite{calders02}.

We will briefly analyze the effect of two common pruning techniques on
discovering statistically significant rules. The important question
is, whether a statistically significant rule can be pruned out as
``redundant'' causing type 2 error.

\subsection{Redundant rules}
\label{redrules}

According to a classical definition (e.g.\ \cite{aggarwalyu2}),
rule $X \rightarrow Y$ is redundant, if there exists $Z \subsetneq X$
such that $fr(X \rightarrow Y)=fr(Z \rightarrow Y)$. The aim of this
definition is to achieve a compact representation of all frequent and
strong association rules. The justification is sensible in the
traditional frequency-confidence-framework with fixed thresholds
$min_{fr}$ and $min_{cf}$: If rule $\ass{Z}{Y}$ is frequent and
strong enough, then all its specializations $\ass{X}{Y}$ with
$P(X,Y)=P(Z,Y)$ are also frequent and strong.

However, this definition is not adequate, if the goal is to find the
most significant rules. In fact, it causes always type 2 error
(rejects the most significant rules), unless $P(X)=P(Z)$. If
$P(X)<P(Z)$, then rule $\ass{X}{Y}$ has higher confidence and is
more significant than $\ass{Z}{Y}$:

\begin{theorem}
If $fr(X \rightarrow  Y)=fr(Z \rightarrow Y)$ for some $Z \subsetneq X$, then 
\begin{itemize}
\item[(i)]$cf(X \rightarrow Y)\geq cf(Z \rightarrow Y)$, and 
\item[(ii)]$cf(X \rightarrow Y)=cf(Z \rightarrow Y)$ only if $P(X)=P(Z)$.
\end{itemize}
\end{theorem}

\begin{proof} 
Let $\xy$ redundant, i.e.\  $\exists Z \subsetneq X$ such that $fr(\xy)=fr(Z
\rightarrow Y)$. Let $X=ZQ$, $P(X)=P(Z,Q)$ and  
$P(X,Y)=P(Z,Y,Q)$. According to the redundancy condition $P(Z,Y,Q)=P(Z,Y)$. 

Now $cf(\xy)-cf(Z \rightarrow Y)=\frac{P(Z,Y,Q)}{P(Z,Q)} - \frac{P(Z,Y)}{P(Z)}$
$=\frac{P(Z,Y)}{P(Z,Q)} - \frac{P(Z,Y)P(Q|Z)}{P(Z,Q)} \geq 0,$
because $P(Q|Z)\leq 1$. $P(Q|Z)=1$ iff $P(Z)=P(Z,Q)=P(X)$.
\hfill $\Box$ 
\end{proof}

Type 1 error (accepting spurious rules) is also likely, because the
least significant rules are output. So, in the worst case all
significant rules are pruned and only spurious rules are presented.

In the context of closed sets, the definition of redundancy is similar
(e.g.\ \cite{zaki}). However, now it is required that there exists
more general rule $\ass{X'}{Y'}$, $X' \subsetneq X$ and $Y' \subseteq
Y$, such that $P(X,Y)=P(X',Y')$ and $P(X)=P(X')$. This means that
$\ass{X}{Y}$ and $\ass{X'}{Y'}$ have the same frequency and
confidence. Still it is possible that
$\gamma(\ass{X}{Y})>\gamma(\ass{X'}{Y'})$ (i.e.\ $P(Y)<P(Y')$) and the
more significant rule is pruned.

\subsection{Productive rules}
\label{prodrules}

According to another common interpretation, rule $X \rightarrow Y$ is
considered redundant or uninteresting, if there exists more general
rule $\ass{Z}{Y}$, $Z \subsetneq X$, such that $P(Y|Z)\geq
P(Y|X)$. Following \cite{webbml} we call these rules {\em
  non-productive}. If $P(Y|X)>P(Y|Z)$ for all $Z \subsetneq X$, rule
$\ass{X}{Y}$ is {\em productive}. The aim of this definition is to
prune out rules which are less interesting than their generalizations.

The heuristic works correctly and avoids type 2 error. For
non-productive rule $\ass{X}{Y}$,
$\gamma(\ass{X}{Y})\leq\gamma(\ass{Z}{Y})$.  In addition, we know that
$P(X,Y)\leq P(Z,Y)$ and $\ass{X}{Y}$ cannot be more significant than
$\ass{Z}{Y}$. In practice, this means that $X$ contains some
attributes which are either independent from $Y$ or negatively
correlated with $Y$.

Generally, it is required that the improvement of
rule $\ass{X}{Y}$ is sufficient \cite{bayardogunopulos}:
\begin{equation}
\label{imp}
  imp(\ass{X}{Y})=cf(\ass{X}{Y})-\max_{Z\subsetneq X}\{cf(Z \rightarrow
  Y)\}\geq min_{imp}. 
\end{equation}
In practice, each rule is compared only to its immediate
generalizations ($|Z|=|X|-1$). If Equation (\ref{imp}) does not hold
for some $Z$, then rule $\ass{X}{Y}$ and all its specializations are pruned.
The problem is that now there could be $X' \supsetneq X$ such that 
$cf(X' \rightarrow Y)>cf(Z \rightarrow Y)$ and which is statistically 
more significant than $\ass{Z}{Y}$. This rule is not discovered,
because the whole branch was pruned. Thus, the pruning condition
should not be used to restrict the search space. 

Instead, the pruning condition can be used in the post-processing
phase, where a rule is compared to all its generalizations.  We will 
show that requirement $min_{imp}=0$ is a necessary but not
sufficient condition for the superiority of $\ass{X}{Y}$ over
$\ass{Z}{Y}$. This means that type 2 error does not occur, but type 1
error (accepting spurious rules) is possible.  However, when
$min_{imp}>0$, also type 2 error is possible, and non-redundant
significant rules can be missed.

The following theorem gives a necessary and sufficient condition for
the superiority of $\ass{X}{Y}$: 

\begin{theorem}
Let $X=Z,Q$ for some $Z,Q \subseteq R$. 
Rule $\ass{X}{Y}$ is more significant than $\ass{Z}{Y}$, if and only if

$$\frac{P(Y|X)-P(Y)}{P(Y|Z)-P(Y)}>\frac{\sqrt{1-P(X)P(Y)}}
{\sqrt{P(Q|Z)(1-P(Z)P(Y))}}.$$
\end{theorem}

\begin{proof} 
\begin{multline*} t(\ass{X}{Y})> t(\ass{Z}{Y}) \Leftrightarrow\\ 
\frac{\sqrt{n}P(X)(P(Y|X)-P(Y))}{\sqrt{P(X)P(Y)(1-P(X)P(Y))}}
>\frac{\sqrt{n}P(Z)(P(Y|Z)-P(Y))}{\sqrt{P(Z)P(Y)(1-P(Z)P(Y))}}\Leftrightarrow\\
\frac{P(Z)P(Q|Z)(P(Y|X)-P(Y))}{\sqrt{P(Z)P(Q|Z)P(Y)(1-P(X)P(Y))}}>
\frac{P(Z)(P(Y|Z)-P(Y))}{\sqrt{P(Z)P(Y)(1-P(Z)P(Y))}}\Leftrightarrow\\
\frac{\sqrt{P(Q|Z)}(P(Y|X)-P(Y))}{\sqrt{(1-P(X)P(Y))}}>
\frac{(P(Y|Z)-P(Y))}{\sqrt{(1-P(Z)P(Y))}}\Leftrightarrow\\
\frac{(P(Y|X)-P(Y))}{(P(Y|Z)-P(Y))}>\frac{\sqrt{(1-P(X)P(Y))}}
{\sqrt{P(Q|Z)(1-P(Z)P(Y))}}\\
\end{multline*}
\hfill $\Box$ 
\end{proof} 

Since $\frac{\sqrt{(1-P(X)P(Y))}}{\sqrt{P(Q|Z)(1-P(Z)P(Y))}}\geq 1,$
it follows that  

\begin{corollary}
If $t(\ass{X}{Y})> t(\ass{Z}{Y})$, then $P(Y|X)>P(Y|Z)$ and 
$imp(\ass{X}{Y})>0$. 
\end{corollary}

Now we can give a better pruning condition than Equation (\ref{imp}): 

\begin{corollary}
If $$\frac{P(Y|X)-P(Y)}{P(Y|Z)-P(Y)}\leq \frac{1}{\sqrt{P(Q|Z)}},$$
then $t(\ass{X}{Y})<t(\ass{Z}{Y})$.
\end{corollary}

The condition can be expressed equivalently as

$$imp(\ass{X}{Y})\leq \frac{(P(Y|Z)-P(Y))(1-\sqrt{P(Q|Z)})}{P(Q|Z)}.$$

This pruning condition is more efficient than
$min_{imp}=0$, but still it does not prune out any non-redundant 
significant rules. Generally, the correct threshold $min_{imp}$
depends on the rules considered and $P(Q|Z)$, and no absolute
thresholds (other than $min_{imp}=0$) can be used.

\section{Conclusions}
\label{concl}
In this paper, we have formalized an important problem: how to find
statistically significant association rules. We have inspected the most
common interest measures and search techniques from the statistical
point of view. For all methods, we have analyzed, whether they can
cause type 1 error (accept spurious rules) or type 2 error (reject
significant rules) and the conditions under which the errors can occur.

The conclusions are the following: The traditional
frequency-confidence framework should be abandoned, because it can
cause both type 1 and type 2 errors. The simplest correction is to
adopt the so called frequency-dependence framework, where the degree
of dependence is used instead of confidence. If the minimum frequency
is set carefully (Theorem \ref{safeminfr}), no significant rules are
missed. On the other hand, all insignificant rules with the desired
level of significance can be pruned, using the $t$-measure.

The $\chi^2$-measure works correctly only if all significant partial
dependencies in the data are actually full dependencies. When it is
used to asses association rules, several spurious rules can be
accepted (type 1 error). Type 2 error does not occur, if the partial
dependencies are sufficiently significant, but the ranking order of 
association rules can be incorrect.

Pearson correlation coefficient $\phi$ is not recommendable for
assessing association rules. It can easily cause both type 1 and type 2
errors. $J$-measure can also cause both error types, although type 1 error
(accepting spurious rules) is more likely. Both $\phi$ and $J$ omit 
the data size, $n$, and it can be hard to decide proper cut-off values for 
significant dependencies. 

Finally, we analyzed two common redundancy reduction techniques, which
compare rule $\ass{X}{Y}$ to its generalizations $\ass{Z}{Y}$,
$Z\subsetneq X$. We showed that the minimum improvement condition,
$imp(\ass{X}{Y})=cf(\ass{X}{Y})-\max_{Z\subsetneq X}\{cf(Z \rightarrow
Y)\}\geq min_{imp}$, works correctly, if $min_{imp}=0$. However, it
cannot be used to restrict the search space, but only for
post-processing. If $min_{imp}>0$, significant rules can be missed. 
We gave also a more efficient pruning condition, which can be used to
prune redundant rules without type 2 error. 

The second redundancy condition, $fr(X \rightarrow Y)=fr(Z \rightarrow
Y)$, does just the opposite and prunes out the more significant,
specific rules. I.e.\ it causes always type 2 error, unless $P(X)=P(Z)$.

In the future research, these new insights should be utilized in the
search algorithms for discovering the statistically most significant
rules. The computational efficiency of such algorithms is a potential
bottle-neck, but the starting point looks promising: while the small
$min_{fr}$ increases the number of frequent sets, we can use 
$\gamma$-based pruning to restrict the search space, and the effects
may compensate each other.

\begin{acknowledgements}
I thank professor O. Kolehmainen for checking the validity of
statistical arguments and professor M. Nyk{\"a}nen for
his valuable comments.
\end{acknowledgements}

\bibliographystyle{spmpsci}      

\end{document}